\documentclass[journal, twoside, web]{ieeecolor}
\usepackage{tikz}
\usepackage{amsmath, amsfonts}
\usepackage{array}
\usepackage{generic}
\usepackage{cite}
\usepackage{graphics, graphicx}
\usepackage{textcomp}
\usepackage{subcaption}
\usepackage{amsfonts}
\usepackage{amsmath}

\usepackage{amsthm}
\usepackage{amssymb}
\usepackage{mathrsfs}
\usepackage{mathtools}
\usepackage[hidelinks]{hyperref}
\usepackage[normalem]{ulem}
\usepackage{bm}
\usepackage{diagbox}
\usepackage{slashbox}
\usepackage{booktabs}
\usepackage{adjustbox}
\usepackage{multirow, multicol}
\usepackage{makecell}
\usepackage{cleveref}
\usepackage{comment}
\usepackage{algorithm}
\usepackage{algpseudocode}
\usepackage{color, colortbl}
\algrenewcommand\textproc{}
\usepackage{threeparttable}
\usepackage{stfloats}
\allowdisplaybreaks

\newcolumntype{M}[1]{>{\centering\arraybackslash}m{#1}}

\theoremstyle{plain}
\newtheorem{lemma}{\textbf{Lemma}}

\newtheorem{theorem}{\textbf{Theorem}}
\newtheorem{corollary}{\textbf{Corollary}}

\newtheorem{problem}{\textbf{Problem}}

\theoremstyle{definition}
\newtheorem{definition}{\textbf{Definition}}
\newtheorem{remark}{Remark}
\newtheorem{example}{\textbf{Example}}
\useunder{\uline}{\ul}{}
\newcommand{\mC}{\mathbb{C}}

\newcommand{\F}{\mathcal{F}}
\newcommand{\Lc}{\mathcal{L}}

\newcommand{\mZ}{\mathbb{Z}}
\newcommand{\Ac}{\mathcal{A}}

\DeclareMathOperator*{\Imag}{Im}
\DeclareMathOperator*{\Real}{Re}

\newcommand{\gc}{\cellcolor[gray]{.90}} 

\makeatletter
\renewcommand{\Function}[2]{%
\csname ALG@cmd@\ALG@L @Function\endcsname{#1}{#2}%
\def\jayden@currentfunction{#1}%
}
\newcommand{\funclabel}[1]{%
\@bsphack \protected@write\@auxout{}{%
\string\newlabel{#1}{{\jayden@currentfunction}{\thepage}}%
}%
\@esphack }
\makeatother

\def\BibTeX{{\rm B\kern-.05em{\sc i\kern-.025em b}\kern-.08em T\kern-.1667em\lower.7ex\hbox{E}\kern-.125emX}}
\begin{document}
  \title{Time-to-reach Bounds for Verification of Dynamical Systems Using the
  Koopman Spectrum}
  \author{Jianqiang Ding and Shankar A. Deka, \IEEEmembership{Member, IEEE}
  \thanks{ Department of Electrical Engineering, Aalto University, Finland.
  Email: {\tt\small \{jianqiang.ding, shankar.deka\}}@aalto.fi.}
  }

  \maketitle

  \begin{abstract}


In this work, we present a novel Koopman spectrum-based reachability verification method for nonlinear systems. Contrary to conventional methods that focus on characterizing all potential states of a dynamical system over a presupposed time span, our approach seeks to verify the reachability by assessing the non-emptiness of estimated time-to-reach intervals without engaging in the explicit computation of reachable set.
Based on the spectral analysis of the Koopman operator, we reformulate the problem of verifying existence of reachable trajectories into the problem of determining feasible time-to-reach bounds required for system reachability. By solving linear programming (LP) problems, our algorithm can effectively estimate all potential time intervals during which a dynamical system can enter (and exit) target sets from given initial sets over an unbounded time horizon.
Finally, we benchmark our method against existing techniques on challenging problems,
such as verifying the reachability between non-convex or even disconnected sets, as well as backward reachability and multiple entries into target sets.
Additionally, we validate its applicability in addressing real-world challenges and scalability to high-dimensional systems through case studies in verifying the reachability of the cart-pole and multi-agent consensus systems.

\end{abstract}

  \begin{IEEEkeywords}
    Koopman operator, reachability analysis, verification, nonlinear dynamical
    systems.
  \end{IEEEkeywords}

\section{Introduction}



As modern society increasingly relies on complex systems in safety-critical areas such as robotics, autonomous driving, and power grids, providing rigorous guarantees to ensure their reliable operation has become crucial. 
Verification of reachability for dynamical systems has been extensively studied within the control community in recent decades. For instance, \cite{asarin2003reachability,asarin2006recent} take an abstraction-based approach wherein reachability analysis is performed on an abstract representation of the original system.
In a different vein, Hamilton-Jacobi (HJ) approaches \cite{bansal2017hamilton}\cite{lygeros2004reachability}
frame the problem as a two-player zero-sum game between the controller and an adversary,
leading to backward reachable set in the form of zero-sublevel set of the solution of a Hamilton-Jacobi partial differential equation, which can be numerically solved with level-set methods \cite{mitchell2007toolbox}. 
Other approaches aim to prove that unsafe states are unreachable by finding invariant sets that include the initial set and exclude the unsafe set, which can be synthesized via optimization techniques involving algebraic geometry \cite{ghorbal2017hierarchy,tiwari2004nonlinear}. 
A popular category of methods in this direction achieve verification by constructing barrier functions, which provide Lyapunov-like guarantees regarding system behavior. The existence of a barrier function is sufficient to conclude the satisfiability of safety or reachability specifications \cite{prajna2007framework,ghaffari2018safety}. 
In addition, set-propagation techniques \cite{althoff2021set} are also worthy of attention.
Starting from initial states, these techniques iteratively compute sets representing all possible behaviors of the system in accordance with its dynamics. 
A variety of tools \cite{althoff2016cora,bogomolov2019juliareach} 
have been developed for reachable set computation in a broad range of dynamical systems, including continuous nonlinear system, hybrid system, and even neural networks.
Yet other methodologies based on satisfiability modulo theory (SMT), seek to verify reachability by searching for
counterexamples.
These approaches has consequently spurred the development of tools like \cite{
kong2015dreach,abate2021fossil} for verification. 
For discrete-time, chance-constrained reachability problems, methods based on moment-sum-of-squares have gained interest recently \cite{wang2020non}.
As discussed in \cite{chen2022reachability}, reachability verification typically relies on computation of finite-time reachable set from an initial set, which is then intersected with the target set. 
Indeed, exploiting system properties such as monotonicity \cite{coogan2020mixed}, 
symmetries \cite{maidens2018exploiting,sibai2020multi}, and decomposition based on system structures \cite{chen2016decomposed}, 
can scale up reachability analysis techniques.
However, this explicit computation of reachable sets, which is often expensive and/or conservative, may not be required in principle since the main goal of verification is to obtain a ``yes" or ``no" answer, and not to obtain actual reachable sets.

\color{black}

\begin{figure*}[htbp!]
    \centering
    \includegraphics[width=0.9\linewidth]{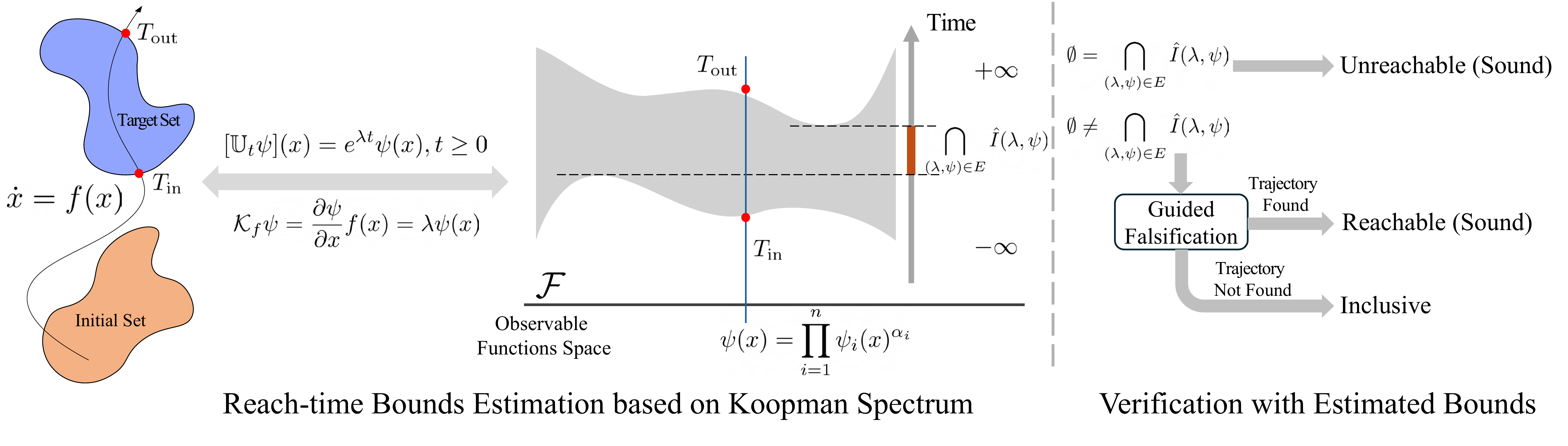}
    \caption{
    Reachability verification pipeline with time-to-reach bounds.
    }
    \label{fig:pipeline}
    \vspace{-1em}
\end{figure*}

Koopman operator theory \cite{koopman1931hamiltonian} provides a powerful framework for globally linearizing nonlinear systems, thus opened new avenues for reachability analysis, and offers a compelling alternative to traditional methods.
This has led to the development of various methodologies for Koopman theory-based reachability analysis.
One of the first algorithms in this direction was proposed in \cite{bak2021reachability}, where the authors verify the reachability of black-box nonlinear systems with Koopman linearization and zonotope over-approximations of nonlinear initial sets.
Subsequently, \cite{bak2022reachability} advance this idea by employing random Fourier features to improve the accuracy of linearization, and combining Taylor models with polynomial zonotope refinement to handle the nonlinear transformation of the initial states.
The reachability of the Koopman bilinear form of control-affine nonlinear systems is discussed in \cite{goswami2021bilinearization}.
In contrast, other methods seek to utilize the Koopman spectrum to analyze the reachability of the system. Similarly, the authors in \cite{umathe2022reachability} use the Koopman operator representation to obtain linear dynamical representation of autonomous nonlinear systems, which aids in reachable set computations which aids in reachable set propagation.

In a similar spirit, 
we propose a novel reachability verification approach based on Koopman spectrum, which eliminates the necessity for explicitly computing reachable sets and focuses on temporal rather than spatial information to verify reachability.
By leveraging the spectrum of the Koopman Operator to decouple reachability verification from the precise system dynamics, our framework gains the crucial flexibility to incorporate diverse Koopman eigenpairs evaluation methods, thereby paving the way for constructing efficient and scalable reachability verification solutions.
The main contributions of this work are summarized as follows:
\begin{itemize}
    \item We introduce a novel Koopman spectrum-based necessary condition for reachability verification. This condition reframes the reachability problem as checking whether the time-to-reach bounds are non-empty, thus avoiding explicit computation of reachable sets and mitigates the dependency on precise modeling of the dynamic.
    \item Building upon this necessary condition, we further relax the estimates of the time-to-reach bounds through parameterization via principal eigenpairs. This relaxation allows the problem can be efficiently solved as linear programming (LP) problems in decision variables that scale linearly with the state-dimensions.
    \item  
    Finally, we propose an reachability verification framework as depicted in Fig. \ref{fig:pipeline}, and demonstrate its practicality through a comparative analysis against existing methodologies on challenging benchmarks, including those with non-convex sets, periodic orbits, and high-dimensional systems.
\end{itemize}\vspace{-0.2cm}
\color{black}
\section{Preliminaries}

\textit{Notations:}
We use $\mathbb{R}$ and $\mathbb{C}$ to denote the set of real and complex numbers, respectively.
$\mathbb{R}^n$ denotes the $n-$dimensional real space. 
The space of all $k-$times continuously differentiable functions on domain $X$ is denoted by $\mathcal{C}^k(X)$.
The notation $\text{spec}(M)$ denotes the spectrum of a matrix $M \in \mathbb{C}^{n \times n}$. $\operatorname{Re}(\cdot)$ and $\operatorname{Im}(\cdot)$ denote the real and imaginary parts of an argument, respectively. $|\cdot|$ is used to denote absolute value of a real or complex number and likewise $\angle \cdot $ denotes the phase angle of a complex scalar.

In this paper, we consider continuous-time dynamical systems of the form
\begin{equation}
\label{definition:dynamic system}
    \frac{d}{dt}x(t) = f(x(t))
\end{equation}

where $f: X \rightarrow \mathbb{R}^n$ is a continuously differentiable map and $x(t) \in \mathbb{R}^n$ denotes the state in the compact set $X \in \mathbb{R}^n$ at time $t \geq 0$. 
The flow map $s_t: X \rightarrow X$ of system \eqref{definition:dynamic system} is given by $s_t(x)=x+\int_0^t f(x(\tau)) d\tau$, for all $t \geq 0$ and $x \in X$. 
We are interested in the following problem:

\begin{problem}
    (Reachability Verification)
    Consider a dynamical system \eqref{definition:dynamic system} with an initial set $X_0 \subset X$ and a target set $X_F \subset X$. Reachability verification aims to determine whether there exists $t \geq 0$ and $x_0 \in X_0$ such that $s_t(x_0) \in X_F$.
\end{problem}
  To tackle this problem, our approach reformulate the reachability verification conditions in terms of Koopman operator's spectrum, 
  relying upon the key definitions presented below.

\begin{definition} (Koopman Operator)
Let $\mathcal{F}$ be a Banach space of scalar-valued functions $\phi(x): X \rightarrow \mathbb{C}$. Then,
the Koopman operator $\mathbb{U}_t: \mathcal{F} \rightarrow \mathcal{F}$ corresponding to the dynamics (\ref{definition:dynamic system}) is defined as
\begin{equation}
    [\mathbb{U}_t \phi](x) = \phi(s_t(x))
\end{equation}
where $\phi(x) \in \mathcal{F}$ commonly refers to as an observable function. We assume $\mathcal{F} \subseteq \mathcal{C}^1(X)$ in this paper.




\end{definition}

\begin{definition} (Koopman Eigenvalues and Eigenfunctions)
    An observable function $\phi(x) \in \mathcal{F}$ is an eigenfunction of the Koopman operator corresponding to the eigenvalue $\lambda$ if
    \begin{equation}
    \label{eq: eigens}
        [\mathbb{U}_t \phi](\textcolor{blue}{x}) = e^{\lambda t}\phi(x), \ t \geq 0.
    \end{equation}

\end{definition}

In the following, we call an eigenfunction $\phi(x)$ as a principal eigenfunction if the corresponding eigenvalue $\lambda$ is also an eigenvalue of the jacobian matrix $\nabla f$ evaluated at equilibrium $x_e.$
Furthermore, we define the set of principal eigenpairs as the minimal generator $G$ of the set given by
    \begin{equation*}
    E = \left\{\left(\sum_{i=1}^{m}n_i\lambda_i,\prod_{i=1}^{m}\phi_i^{n_i}\right) \bigg\lvert \;(\lambda_i,\phi_i) \in G,\; n_i \in \mathbb{N}     \right\}.
    \end{equation*}
    where $E$ is the semigroup of eigenpairs $(\lambda, \phi)$.  
The concept of principal eigenfunctions was introduced in \cite{mohr2016koopman}, and uniqueness and existence of these principal eigenfunctions have been studied rigorously in \cite{kvalheim2021existence}. Due to its principal algebraic structure, the cardinality of set $E$ is countably infinite. As pointed out in \cite{bollt2021geometric}, this set doesn't contain all the possible Koopman eigenfunctions, and the authors consequently introduced the concept of `primary eigenfunctions,' to allow the exponents $n_i$ in the definition of set $E$ above to be real valued. We shall follow this extension by \cite{bollt2021geometric} for parameterizing the Koopman spectrum in our paper.



\section{Time-to-reach bounds}

Consider a set of initial conditions $X_{0}\subset X$ and a target final set
$X_{F}\subset X$ such that $x_{e}\notin X_{0}\cup X_{F}$. Additionally, given
any function $g: X \rightarrow 
\mC
$ and a set $V,W \subset X$, we define the following notations for convenience:
\begin{gather}
    \overline{g}(V) \doteq \sup_{x\in V}
    |g(x)|
    ,\quad \quad \underline{g}(V) \doteq
    \inf_{x\in V}
    |g(x)|
    \label{eq:def_over_under}\\ \Lc^{g}(W,V) \doteq \log\left(\frac{\overline{g}(V)}{\underline{g}(W)}
    \right).
\end{gather}
\noindent
Let us consider any Koopman eigenfunction 
$\psi \in
\F$
, corresponding to eigenvalue 
$\lambda \in \mC$
. We
can now state the following theorem on time-to-reach bounds.

\begin{theorem}\label{th:necessary_real}
    Given a dynamical system \eqref{definition:dynamic system}, let $(\lambda,\psi)$ be any non-trivial Koopman eigenpair defined over the set $X$. If a set $X_F$ is reachable from a set $X_0$, then the time-to-reach, T, satisfies the following bounds:
    \begin{equation}\label{eq:real_T}
        \Real{(\lambda)}T \in \bigg[-\Lc^{\psi}(X_{F},X_{0}) \, , \Lc^{\psi}(X_{0}
        ,X_{F}) \bigg].
    \end{equation}
   
\end{theorem}

\begin{proof}
    Let $x_{0}$ be some point in $X_{0}$ such that $s_{T}(x_{0})\in X_{F}$ for some
    $T>0$. By definition, $\psi(s_{T}(x_{0})) = e^{\lambda T}\psi(x_{0})$, which implies $|\psi(s_{T}(x_{0}))| = e^{\Real{(\lambda)} T}|\psi(x_{0})|$. 
    We note that $\underline{\psi}(X_F) \le |\psi(s_{T}(x_{0}))| \le \overline{\psi}(X_F)$ and likewise, $\underline{\psi}(X_0) \le |\psi(x_{0})| \le \overline{\psi}(X_0)$.  These inequalities lead to
    \begin{equation*}
        e^{\Real{(\lambda)} T}\underline{\psi}(X_0) \le \overline{\psi}(X_F)\,
        \text{ and }   \underline{\psi}(X_F) \le e^{\Real{(\lambda)} T}\overline{\psi}(X_0).
    \end{equation*}
    Taking natural logarithm on both sides and rearranging the terms
    leads us to equation \eqref{eq:real_T}, thus completing the proof. \end{proof}
We shall denote the reach-time interval containing $T$ (obtained by moving $\Real{(\lambda)}$ in equation \eqref{eq:real_T} to the right-hand side) as $I_{mag}(\lambda,\psi).$ 

In case of complex eigenfunctions, one can utilize information encoded in the phase
(or argument) of the eigenfunctions and the imaginary part of the corresponding eigenfunction in addition to its modulus, for further refining the time intervals
$I(\lambda,\psi)$. For brevity of presentation, let us introduce the notation
$\Ac^{g}(W,V) \doteq \overline{\angle g}(V) - \underline{\angle g}(W)$ given any
complex function $g: X \rightarrow \mathbb{C}$ and sets $V,W \subset X$. Thus, we
present the following result for reach bounds in terms of the Koopman
eigenfunction phase.

\begin{theorem}
    \label{th:necessary_imag} 
    Consider again the dynamical system \eqref{definition:dynamic system}, along with a non-trivial Koopman eigenpair $(\lambda,\psi
    )$ over $X$
    . If a set $X_{F}$ is reachable from an initial set $X
    _{0}$, then the time-to-reach, $T$, would necessarily satisfy
    \begin{equation}\label{eq:imag_T}
        \Imag{(\lambda)}T \in \bigg[-\Ac^{\psi}(X_{F},X_{0}) \, , \Ac^{\psi}(X_{0}
        ,X_{F}) \bigg] + 2m\pi
    \end{equation}
    for some $m\in\mathbb{Z}.$
\end{theorem}
\begin{proof}
    For any $x\in X_{0}$ and $t\ge 0$, we have
    \begin{eqnarray*}
        \psi(s_t(x)) = e^{\lambda t}\psi(x) = |\psi(x)|e^{\Real{(\lambda) t}}e^{j\Imag{(\lambda)} t + j\angle \psi(x)}\\
        \implies \angle \psi(s_t(x)) = \angle\psi(x) + \Imag{(\lambda)}t - 2m\pi
    \end{eqnarray*}
    for some $m \in \mZ$. Now, if $\exists T>0$ such that $s_{T}(x)\in X_{F}$, then we have $\overline{\angle \psi}(X_F) \ge \angle \psi(s_T(x)) = \angle\psi(x) + \Imag{(\lambda)}T
        - 2m\pi$, which greater than or equal to $\underline{\angle\psi}(X_0) + \Imag{(\lambda)}T - 2m\pi$. Rearranging, we get $\Imag{(\lambda)}T \le \overline{\angle \psi}(X_F) -
        \underline{\angle\psi}(X_0) + 2m\pi.$
    Similarly, $\underline{\angle \psi}(X_F) \le \angle \psi(s_T(x)) = \angle\psi(x) +
        \Imag{(\lambda)}T - 2m\pi$, which is less than or equal to $\overline{\angle\psi}(X_0) + \Imag{(\lambda)}T
        - 2m\pi.$ This leads to $\Imag{(\lambda)}T \ge \underline{\angle \psi}(X_F)
        - \overline{\angle\psi}(X_0) + 2m\pi.$
    This completes the proof.
\end{proof}
We shall denote the time-to-reach interval derived from equation \eqref{eq:imag_T}, by taking $\Imag{(\lambda)}$ to the right-hand side of the equation, as
$I_{phase}(\lambda,\psi)$. Note that unlike $I_{mag}(\lambda,\psi)$, the set $I_{phase}(\lambda,\psi)$ is a collection of intervals separated by a period of $\frac{2\pi}{\Imag{(\lambda)}}$. We define $I(\lambda,\psi) \doteq I_{mag}(\lambda,\psi) \cap I_{phase}(\lambda,\psi)$, and can now present the following result.

\begin{corollary}\label{cor:main}
    For the dynamical system \eqref{definition:dynamic system}, a necessary condition for a target set $X_{F}  \subset X$ to be reachable from an initial set $X_{0} \subset X$ is that intersection of all intervals $I(\lambda,\psi)$ is non-empty
    . In other
    words,
        \begin{gather*}
            \left\{s_t(x_0) \,|\, x_0 \in X_0 \right\}
            \bigcap X_F \neq \emptyset \text{ for some } t>0 \\
            \implies \bigcap_{(\lambda,\psi)\in E} I(\lambda,\psi) \neq \emptyset.
        \end{gather*}
\end{corollary}

\subsection{Time-bounds parameterized through principal eigenfunctions}
\label{sub:principle_bounds} 
Verifying reachability using corollary \ref{cor:main}
involves taking intersections of (possibly) an infinite number of intervals.
Towards making this procedure conducive to practical implementation, 
we relax the estimation of $I(\lambda,\psi)$ by constructing its over-approximation $\hat{I}(\lambda,\psi)$, which can be parameterized via a finite number of principal eigenpairs.


\begin{theorem}
    \label{th:relaxation} 
    Let $I_{mag}(\lambda,\psi)$ be defined as in theorem \ref{th:necessary_real} for eigenpair $(\lambda,\psi
    )\in E$, such that
    $\psi = \prod_{i=1}^{n}\psi_{i}^{\alpha_i}$ and
    $\lambda = \sum_{i=1}^{n}\alpha_{i}\lambda_{i}$ with $\alpha_{i}\ge 0$ and principal eigenpairs $(\lambda_i,\psi_i)\in G$
    . Then, we have
    \begin{enumerate}
        \item[(a)] For $\Real{(\lambda)}>0$, we have $I_{mag}(\lambda,\psi) \subseteq \hat{I}_{mag}(\lambda,\psi) \doteq$ 
    \[ \hspace{0.8em}\left[\frac{\sum_{i=1}^{n}\alpha_{i}\Lc^{\psi_i}(X_{F},X_{0})}{-\sum_{i=1}^{n}\alpha_{i}\Real{(\lambda_{i})}}
        , \frac{\sum_{i=1}^{n}\alpha_{i}\Lc^{\psi_i}(X_{0},X_{F})}{\sum_{i=1}^{n}\alpha_{i}\Real{(\lambda_{i})}}
        \right] 
    \] 

    \item[(b)] For $\Real{(\lambda)}<0$, we have $I_{mag}(\lambda,\psi) \subseteq \hat{I}_{mag}(\lambda,\psi) \doteq$
    \[
        \hspace{0.8em} \left[\frac{\sum_{i=1}^{n}\alpha_{i}\Lc^{\psi_i}(X_{0},X_{F})}{\sum_{i=1}^{n}\alpha_{i}\Real{(\lambda_{i})}}
        , \frac{\sum_{i=1}^{n}\alpha_{i}\Lc^{\psi_i}(X_{F},X_{0})}{-\sum_{i=1}^{n}\alpha_{i}\Real{(\lambda_{i})}}
        \right]
    \]
    \end{enumerate}
\end{theorem}
\begin{proof}
    Since logarithm is a monotonically increasing function, we can swap the
    order of $\log$ and $\sup$, such that for any set $S \subseteq X$, we have $\log \sup_{x\in S}|\psi(x)| = \sup_{x\in S} \log |\psi(x)| = \sup_{x\in S}
        \sum_{i=1}^n \alpha_i \log |\psi_i(x)|.$ Applying triangular inequality, the last term is bounded by $\sum_{i=1}^n \alpha_i \sup_{x\in S}
        \log |\psi_i(x)| = \sum_{i=1}^n \alpha_i \log \sup_{x\in S}|\psi_i(x)|.$
    In other words,\vspace{-1em}
    \begin{equation}
        \label{eq:sup_relaxed}\log \overline{\psi}(S) \le \sum_{i=1}^{n}\alpha_{i}
        \log \overline{\psi_i}(S). \vspace{-1em}
    \end{equation}
    \noindent
    Similarly, we can show that \vspace{-1em}
    \begin{equation}
        \label{eq:inf_relaxed}\log \underline{\psi}(S) \ge \sum_{i=1}^{n}\alpha_{i}
        \log \underline{\psi_i}(S). \vspace{-1em}
    \end{equation}
    Now using equations \eqref{eq:sup_relaxed} and \eqref{eq:inf_relaxed} in equation
    \eqref{eq:real_T}, we get the following when $\Real{(\lambda)}>0$:
    \begin{align*}
              &-\frac{1}{\Real{(\lambda)}}\Lc^\psi(X_F,X_0) = \frac{1}{\Real{(\lambda)}}\log\left(\frac{\underline{\psi}(X_{F})}{\overline{\psi}(X_{0})}\right) \notag\\
              &\ge  \frac{\sum_{i=1}^{n}\alpha_{i}\log \left(\frac{\underline{\psi_i}(X_F)}{\overline{\psi_i}(X_0)}\right)}{\sum_{i=1}^{n}\alpha_{i}\Real{(\lambda_{i})}} =\frac{\sum_{i=1}^{n}\alpha_{i}\Lc^{\psi_i}(X_{F},X_{0})}{-\sum_{i=1}^{n}\alpha_{i}\Real{(\lambda_{i})}}, \text{ and} \\
        \notag 
              & \frac{1}{\Real{(\lambda)}}\Lc^{\psi}(X_{0},X_{F}) = \frac{1}{\Real{(\lambda)}}\log\left(\frac{\overline{\psi}(X_{F})}{\underline{\psi}(X_{0})}\right) \notag   \\
              & \le \frac{\sum_{i=1}^{n}\alpha_{i}\log \left(\frac{\overline{\psi_i}(X_F)}{\underline{\psi_i}(X_0)}\right)}{\sum_{i=1}^{n}\alpha_{i}\Real{(\lambda_{i})}}= \frac{\sum_{i=1}^{n}\alpha_{i}\Lc^{\psi_i}(X_{0},X_{F})}{\sum_{i=1}^{n}\alpha_{i}\Real{(\lambda_{i})}}. 
    \end{align*}
    In the same way, for the case of $\Real{(\lambda)} < 0$, we have
    \begin{align*}
        -\frac{1}{|\Real{(\lambda)}|}\Lc^{\psi}(X_{0},X_{F}) & \ge \frac{\sum_{i=1}^{n}\alpha_{i}\log \left(\frac{\underline{\psi_i}(X_0)}{\overline{\psi_i}(X_F)}\right)}{-\sum_{i=1}^{n}\alpha_{i}\Real{(\lambda_{i})}}, \text{ and}    \\
        \notag                                       
        \frac{1}{|\Real{(\lambda)}|}\Lc^{\psi}(X_{F},X_{0})  & \le \frac{\sum_{i=1}^{n}\alpha_{i}\log \left(\frac{\overline{\psi_i}(X_0)}{\underline{\psi_i}(X_F)}\right)}{-\sum_{i=1}^{n}\alpha_{i}\Real{(\lambda_{i})}}. 
    \end{align*}
    This completes our proof.
\end{proof}
\begin{remark}
Theorem \ref{th:relaxation} defines an interval $\hat{I}_{mag}(\lambda,\psi)$ paramterized by the principal eigenpairs, which provides an over-approximation of $I_{mag}(\lambda,\psi)$. Note this over-approximation is tight and equality occurs, for example, when the sets $X_0$ and $X_F$ are expressed in the form of sublevel sets and superlevel sets of the principal eigenfunctions i.e., as $\left\{ x \in X\, | \, \underline{\gamma_i} \le \psi_i(x) \le \overline{\gamma_i},\, i=1,\ldots,n\right\}.$ This is because $\psi(x)$ can then be maximized/minimized over the sets $X_0$ and $X_F$ by individually maximizing/minimizing the principal eigenfunctions, turning inequalities \eqref{eq:sup_relaxed} and \eqref{eq:inf_relaxed} into equalities.
\end{remark} 

Just like in the case of 
$I_{mag}(\lambda,\psi)$, one can obtain a tight, finite-dimensional over-approximation $\hat{I}_{phase}(\lambda,\psi)$ of the interval $I_{phase}(\lambda,\psi)$
as follows. 

\begin{theorem}
    \label{th:complex_relaxation} Given a complex (and non-trivial) eigenpair $(\lambda
    ,\psi)\in E$, parameterized by $\psi = \prod_{i=1}^{n}\psi_{i}^{\alpha_i}$
    and $\lambda = \sum_{i=1}^{n}\alpha_{i}\lambda_{i}$, where $\alpha_{i}\ge 0$,
    we have
    
    \noindent
    \begin{enumerate}
        \item [(a)] For $\Imag{(\lambda)}> 0, \text{ and some }m \in \mathbb{Z}$,\vspace{-1em}
    \begin{gather*}
        \hspace{0.8em}I_{phase}(\lambda,\psi) 
        \subseteq
        \Bigg[\frac{\sum_{i=1}^{n}\alpha_{i}\Ac^{\psi_i}(X_{0},X_{F})
        + 2m\pi}{\sum_{i=1}^{n}\alpha_{i}\Imag{(\lambda_i)}}\; , \hspace{8em}\\ \hspace{2em}
        \frac{\sum_{i=1}^{n}-\alpha_{i}\Ac^{\psi_i}(X_{F},X_{0}) + 2m\pi}{\sum_{i=1}^{n}\alpha_{i}\Imag{(\lambda_i)}}
        \Bigg] \doteq \hat{I}_{phase}(\lambda,\psi) 
    \end{gather*}

    \item[(b)] For $\Imag{(\lambda)}< 0, \text{ and some }m \in \mathbb{Z}$,
    \begin{gather*}
        \hspace{0.8em}I_{phase}(\lambda,\psi) 
        \subseteq
        \Bigg[\frac{\sum_{i=1}^{n}\alpha_{i}\Ac^{\psi_i}(X_{F},X_{0})
        +2m \pi}{-\sum_{i=1}^{n}\alpha_{i}\Imag{(\lambda_i)}}\; , \hspace{8em}\\
        \hspace{2em}\frac{\sum_{i=1}^{n}\alpha_{i}\Ac^{\psi_i}(X_{0},X_{F})+2m
        \pi}{\sum_{i=1}^{n}\alpha_{i}\Imag{(\lambda_i)}}\Bigg] \doteq \hat{I}_{phase}(\lambda,\psi) 
    \end{gather*}
    \end{enumerate}
\end{theorem}
\begin{proof}
    We begin by noting that $\angle \prod_{i=1}^{n}\psi_{i}^{\alpha_i}= \sum_{i=1}
    ^{n}{\alpha_i}\angle \psi_{i}+ 2k\pi$ for some $k\in \mathbb{Z}$. For any
    set $S \subseteq X$, this gives us
    \begin{eqnarray*}\allowdisplaybreaks
        &&\overline{\angle \psi}(S) = \sup_{x\in S}\angle \psi(x) = \sup_{x\in S}
        \sum_{i=1}^{n}{\alpha_i}\angle \psi_i(x) + 2k\pi \underset{\begin{array}{c}\text{\scriptsize triangular} \vspace{-0.5em}\\ \text{\scriptsize inequality}\end{array}}{\le}\\ 
        &&\sum_{i=1}^{n} \sup_{x\in S} {\alpha_i}\angle \psi_i(x) + 2k\pi = \sum_{i=1}^{n}
        \alpha_i \overline{\angle \psi_i}(S) + 2k\pi, \\
        && \text{ and similarly } \underline{\angle \psi}(S)
         \underset{\begin{array}{c}\text{\scriptsize triangular} \vspace{-0.5em}\\ \text{\scriptsize inequality}\end{array}}{\ge}
        \sum_{i=1}^{n} \alpha_i \underline{\angle \psi_i}(S) + 2k\pi.
    \end{eqnarray*}
    Therefore,
    \begin{eqnarray*}
        \Ac^\psi(X_0,X_F) &\le& \sum_{i=1}^{n} \alpha_i\overline{\angle \psi_i}(X_F)
        - \sum_{i=1}^{n} \alpha_i\underline{\angle \psi_i}(X_0) \text{ and }\\ -\Ac^\psi(X_F,X_0)
        &\ge& \sum_{i=1}^{n} \alpha_i\underline{\angle \psi_i}(X_F) - \sum_{i=1}^{n}
        \alpha_i\overline{\angle \psi_i}(X_0)
    \end{eqnarray*}
    The rest of the proof follows simply by using these above inequalities in
    the expression for $I_{phase}
    (\lambda,\psi)
    $ obtained from theorem \ref{th:necessary_imag}
    for the two cases ($\Imag{(\lambda)}<0$ and $\Imag{(\lambda)}>0$).
\end{proof}

\begin{remark}
    Note that, although the motivation of our algorithm stems from estimating
    reach-time bounds, the upper limit of the interval $\cap_{(\lambda,\psi)\in E} I(\lambda,\psi)$ also represents the latest time the system
    remains within the target sets. In other words, our method estimates both the
    reach-time and the duration within the target sets, as illustrated in the example below.
\end{remark}

\begin{example}
    (Duffing's oscillator) Consider the nonlinear dynamics
    \[
        \left[
        \begin{array}{c}
            \dot{x}_1 \\
            \dot{x}_2
        \end{array}\right] = \left[
        \begin{array}{c}
            {x}_2                  \\
            -0.5x_2 - x_1(x_1^2-1)
        \end{array}\right],
    \]
    with stable equilibrium points at $(\pm 1,0)$ and a saddle equilibrium point at $(0,0)$. We start with non-convex initial and target sets given by $X_{0}=h_{-0.75,1.85,1,2,0.05}
    (x)\leq-0.01$ and $X_{F}=h_{0.65,0.25,3,4,0.1}(x)\leq-0.9$, as the initial
    set and target set, respectively, where
    $h_{x_1^c,x_2^c,a,b,s}=-(1-\frac{x_{1}-x_{1}^{c}}{3}+a(\frac{x_{2}-x_{2}^{c}}{s}
    )^{5}+b(\frac{x_{1}-x_{1}^{c}}{s})^{3}) \cdot e^{-((\frac{x_{1}-x_{1}^{c}}{s})^2+(\frac{x_{2}-x_{2}^{c}}{s})^2)}$.
    By taking the intersection of the estimates of 
$\cap_{(\lambda,\psi)\in E} \hat{I}_{mag}(\lambda,\psi) = [3.83, 7.64]$, and $\cap_{(\lambda,\psi)\in E} \hat{I}_{phase}(\lambda,\psi) = [-3.98,-3.19]$ with a period $P=4.52$, 
    we get our final reach time interval estimate as $[5.05,5.83]$.
    The simulation result in Fig. \ref{fig:ex30} validate this estimation.

\begin{figure}[htbp!]
    \centering
    \begin{subfigure}[t]{0.49\linewidth}
    \centering
\includegraphics[width=\linewidth]{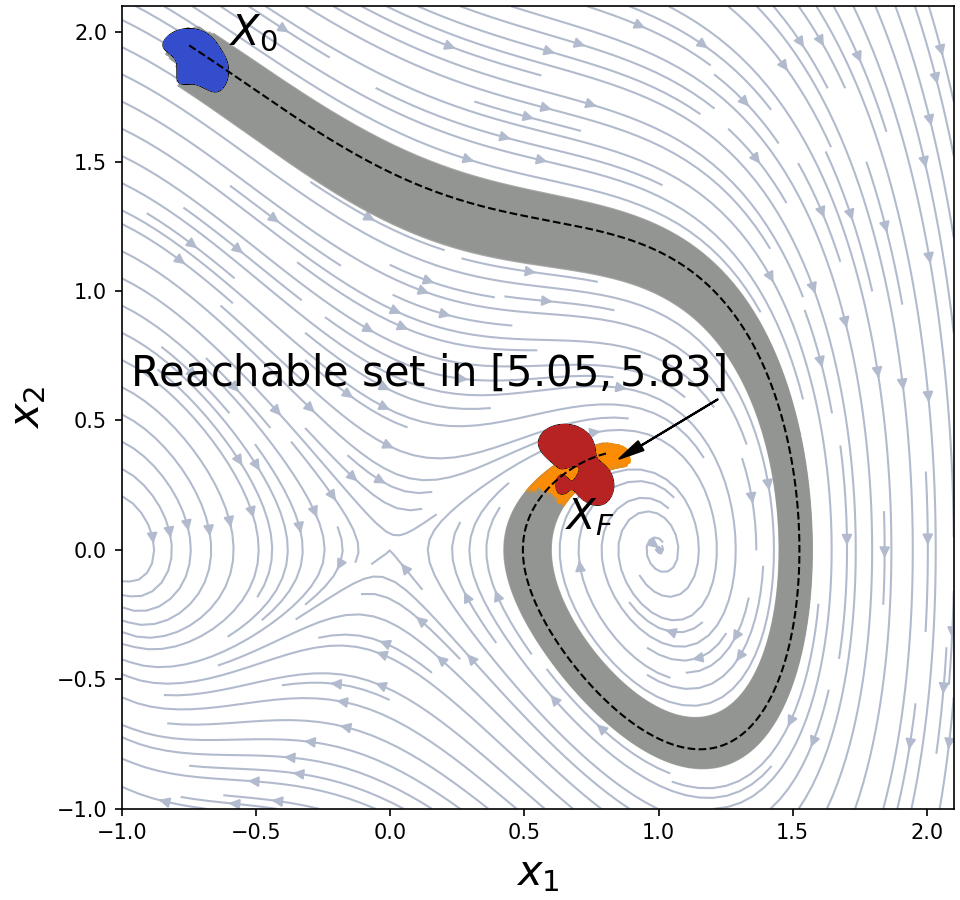}
    \caption{
    Reachable sets of the Duffing's oscillator system from non-convex
    initial set $X_{0}$ for estimated reach-time bounds $I=[5.05, 5.83]$ (orange) with a simulated reachable trajectory (black dashed line), the full reachable set (gray) in $[0,5.84]$ (gray) shown for reference.
    }
    \label{fig:ex30}
    \end{subfigure}
    \hfill
    \begin{subfigure}[t]{0.49\linewidth}
    \centering
\includegraphics[width=\linewidth]{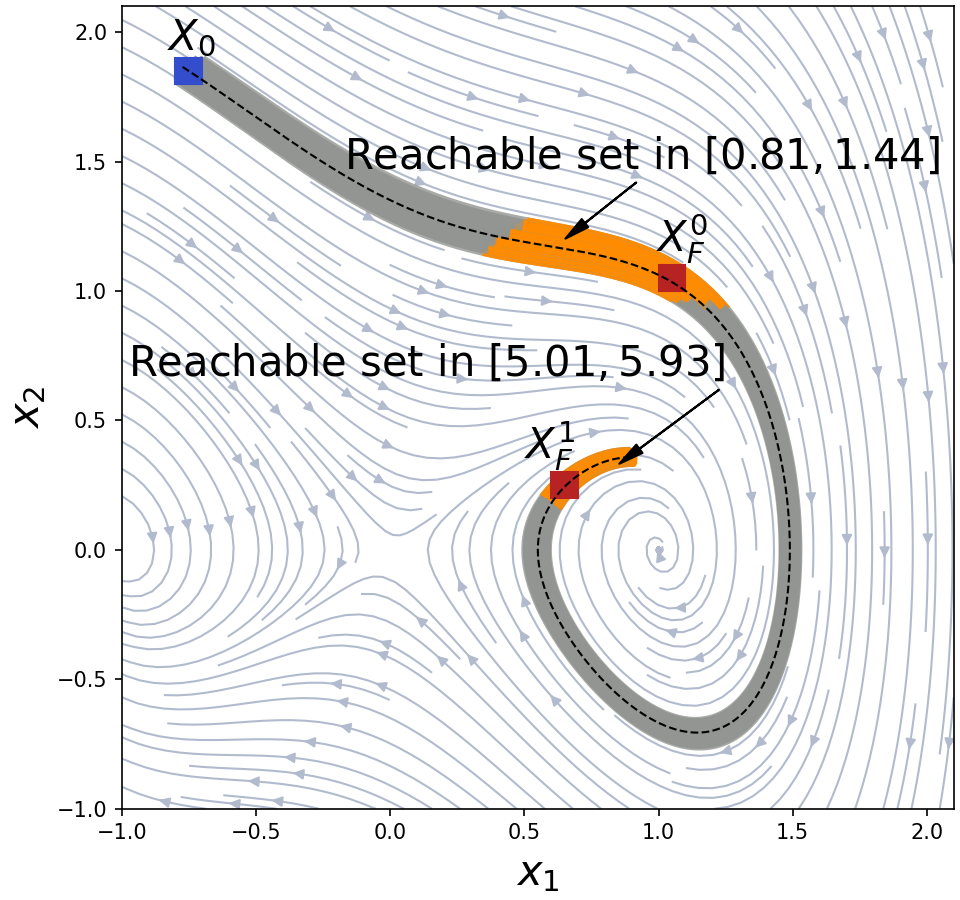}
    \caption{
    Reachable set of the Duffing's oscillator system from initial
    set $X_{0}$ to reach a disjoint target set
    $X_{F}=X_{F}^{0} \bigcup X_{F}^{1}$ for estimated reach-time bounds $I_{0}
    =[0.82, 1.44]$ (orange), $I_{1}=[5.01, 5.93]$ (orange) with a
    simulated reachable trajectory (black dashed lines). The full reachable set for duration $[0,5.93]$ provided for reference.
    }
    \label{fig:ex31}
    \end{subfigure}
    \caption{Reach-time bounds estimation for the Duffing's oscillator system.}
    \vspace{-1em}
\end{figure}
    Next, we consider an initial set $X_{0}=[-0.8,-0.7]\times[1.8,1.9]$ and a target
    set $X_{F}=X_{F}^{0} \bigcup X_{F}^{1} = [1.0,1.1]\times[1.0,1.1] \bigcup [0.
    6,0.7]\times[0.2,0.3]$ which is the union of two disjoint sets. Our method
    provides the estimate of the reachable time as a union of two disjoint
    intervals, $I_{0}=[0.815,1.443]$ and $I_{1}=[5.009,5.934]$, indicating the
    possible existence of trajectories from $X_{0}$ that may pass through the
    disjoint target sets twice. Specifically, the first reach time to the target set
    ranges from $0.815$ to $1.443$, and the trajectories may re-enter the target
    set as early as $5.009$ and exit no later than $5.934$. The simulated
    trajectories in Fig. \ref{fig:ex31} corroborate these estimates.
    
    \textbf{Robustness to eigenfunction estimation errors:} The final computed reach-time bounds can be distorted by inaccuracies in principal eigenpairs. 
    We demonstrate this
    in our Duffing's example by considering the inaccurate principal eigenfunction $\tilde{\psi_i} = (1+\delta(x)) \psi_i(x)$ with state-dependent multiplicative noise for the nonconvex case, where $\psi_i(x)$ refers to the exact principal eigenfunction of the system, and $\delta(x) \sim U(-\epsilon,\epsilon)$ is a random variable following uniform distribution. 
    Fig. \ref{fig:bound with noise} illustrates the distribution of the estimated reach-time bounds obtained by our method using eigenfunction modulus under varying levels of inaccuracy of the principal eigenfunctions.
    \begin{figure}[htbp!]
        \centering
        \includegraphics[width=\linewidth]{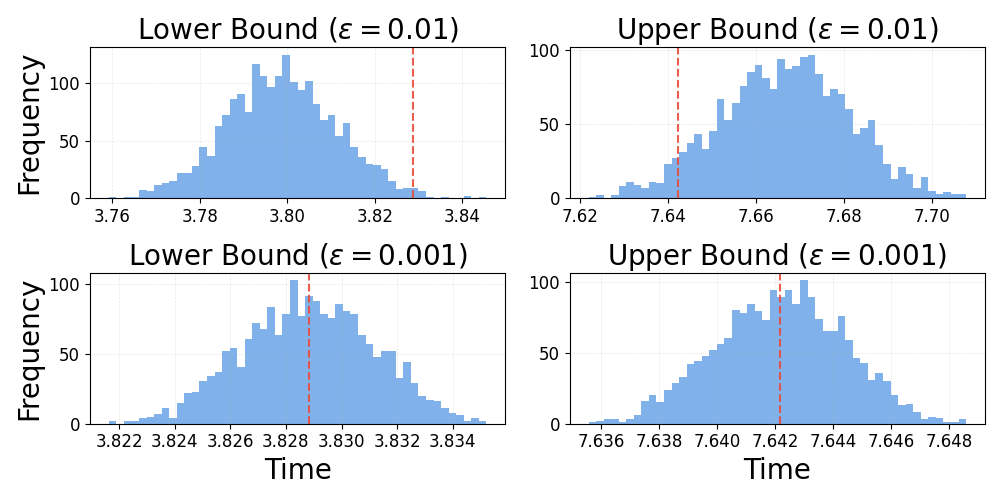}
        \caption{
        Distribution of estimated reach-time bounds $\hat{I}_{mag} (\lambda, \psi)$ with inaccurate principal eigenfunctions. Red dashed lines indicate error-free case, i.e. when $\delta(x) \equiv 0$.}
        \label{fig:bound with noise}\vspace{-1em}
    \end{figure}
\end{example}
\subsection{Numerical implementation}
As previously mentioned in theorems \ref{th:relaxation} and
\ref{th:complex_relaxation}, we can estimate the upper and lower bounds of the reachable
time interval, by exploring all possible
$\alpha_{i}\in [0,+\infty), i=1,\ldots,n$, which parameterize all Koopman eigenfunctions
and eigenvalues in terms of the principal Koopman spectrum. For reachable time
bounds derived from the real parts of principal eigenfunctions, the determination
of each lower (upper) bound essentially defines a linear-fractional programming (LFP)
problem to maximize (minimize) a ratio of affine functions as follows.
\begin{equation}
    \max_{z \in \mathbb{R}^n}\;\frac{c^\top z+e}{d^\top z+f} \quad\text{s.t.
        } Az \leq b, \tag{\textsc{Problem I}}
\end{equation}
where $c,d \in \mathbb{R}^{n}$, $b \in \mathbb{R}^{m}$, $A \in \mathbb{R}^{m
\times n}$ and $e,f\in \mathbb{R}$ are constants. By applying the Charnes-Cooper
transformation, this LFP problem can be effectively solved as the following linear
programming (LP) problem:
\begin{equation}
\max_{y \in \mathbb{R}^n}\;c^\top y+et \quad \text{s.t.}
\Biggl\{\begin{array}{c}
     Ay \leq bt,  \\
     d^\top y + ft =1, \\
        t \geq 0.
\end{array} \quad
\tag{\textsc{Problem II}}
\end{equation}
\begin{lemma}[In \cite{cooper1962programming}, page 183]
    If $z^{*}$ is an optimal solution of \textsc{problem I} such that
    $d^\top z^{*}+ f > 0$, and if $(y^{*},t^{*})$ is an optimal solution of \textsc{problem
    II}, then $y^{*}/t^{*}$ is an optimal solution of \textsc{problem I}.
\end{lemma}

\color{black}

As indicated in theorem \ref{th:complex_relaxation},
the estimated reach-time bounds based on the phase information can be
periodically distributed. 
Our goal is to identify the reach-time
bound for one cycle and the corresponding period.
To achieve a feasible tightest estimate of one reachable time range within a single
period (in the case of $\Imag{(\lambda)}>0$), we solve the following
optimization problem:
\begin{equation*}
    \arg \min_{\alpha_i} \;\frac{\sum_{i=1}^{n}\alpha_{i}c_{i}}{\sum_{i=1}^{n}\alpha_{i}d_{i}} \quad
    \text{s.t.} \Bigg\{        
    \begin{array}{c}
         \sum_{i=1}^{n}\alpha_{i}d_{i}>0,  \\
         \alpha_{i}>0,
    \end{array}
\end{equation*}
where
$c_{i}= \overline{\angle \psi_i}(X_{F}) - \underline{\angle \psi_i}(X_{F}) + \overline
{\angle \psi_i}(X_{0}) - \underline{\angle \psi_i}(X_{0})$
and $d_{i}= \Imag{(\lambda_i)}$. Given the optimal solution $\alpha^{*}=\{\alpha_{0}
^{*}, \cdots, \alpha_{n}^{*}\}$, the period of the reachable time bounds can be
determined as $P = \frac{2\pi}{\sum_{i=1}^{n}\alpha_{i}^{*}\Imag{(\lambda_i)}}$.
Computations for the case of $\Imag{(\lambda)}<0$ are done similarly.

\color{black}

\section{Numerical experiments}

In this section, we test our proposed method in addressing challenging
reachability verification problems, aiming to demonstrate its broad applicability
in the field of reachability analysis. 
Except for the benchmarks on system NL-SLC and NL-EIG
where the principal eigenfunctions are
known in closed-form, the values of principal eigenfunctions over given sets in all
other examples are numerically determined using the path-integral formula \cite{deka2023path}.
Note that our method is independent of the specific choice of estimation
technique for obtaining the principal eigenfunctions and one could use, for example,
Extended-DMD or other data-driven methods in the absence of  dynamic model.
All programs are run on an Intel Core i5-13400 CPU with 32GB RAM and solved using
\textsc{Mosek}. The computational time for each example are listed in the Table
\ref{tab:benchmark}.

\subsection{Applications}

\begin{example}
    \label{ex7} (Cart-pole system \cite{tedrake2009underactuated}) Consider the dynamic
    of a cart-pole system as follows:
    \begin{align*}
        \dot{x}= \begin{bmatrix}\dot{p} \\ \dot{v} \\ \dot{\theta} \\ \dot{\omega}\end{bmatrix} = \begin{bmatrix}v \\ \frac{[F+m_{p}\sin{(\theta)}(l\omega^{2}+g\cos{(\theta)})]}{m_{c}+m_{p}\sin{(\theta)}^{2}} \\ \omega \\ \frac{[-F\cos{(\theta)}-m_{p}l\omega^{2}\cos{(\theta)}\sin{(\theta)}-(m_{c}+m_{p})g\sin{(\theta)}]}{l(m_{c}+m_{p}\sin{(\theta)}^{2})}\end{bmatrix}
    \end{align*}
    where $m_{c}=2, m_{p}=1, l=1$, and $g=9.8065$.


    \begin{figure}
        \centering
        \subfloat[]{ \includegraphics[width=\linewidth]{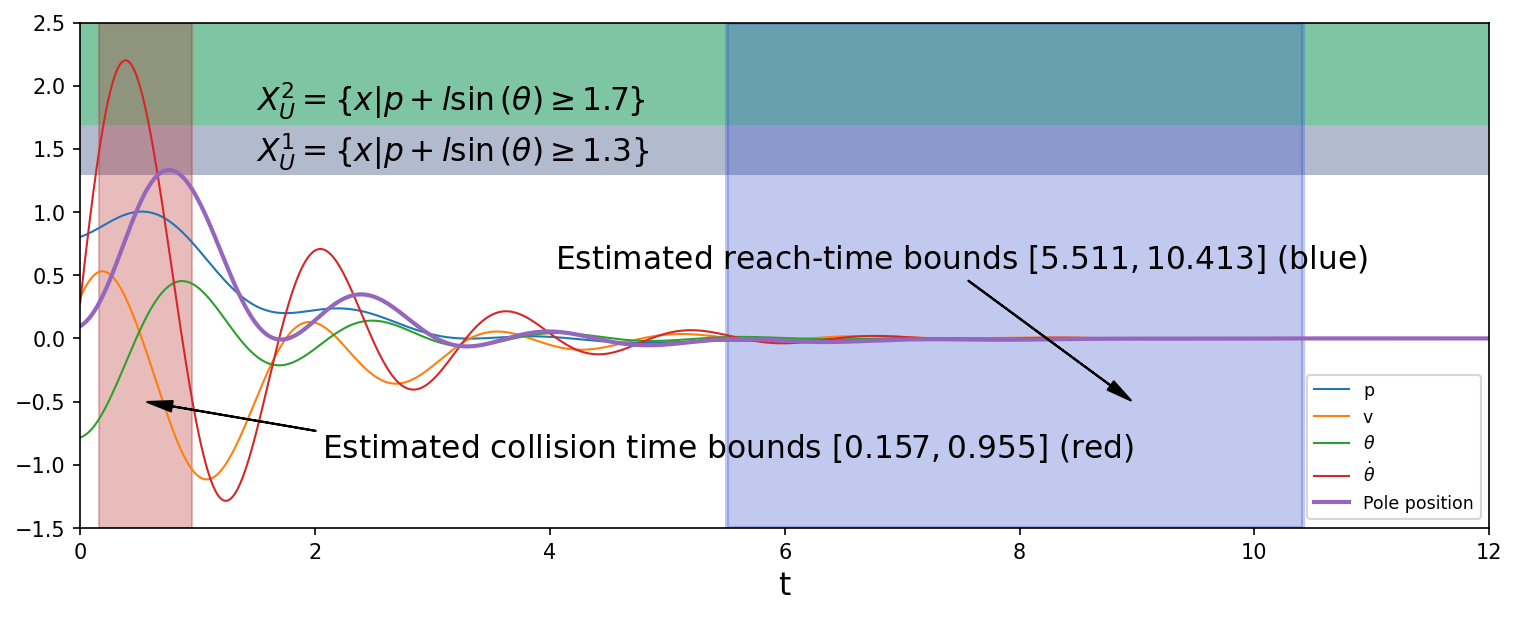} \label{fig:ex60} }
        \\ \subfloat[]{ \includegraphics[width=0.48\linewidth]{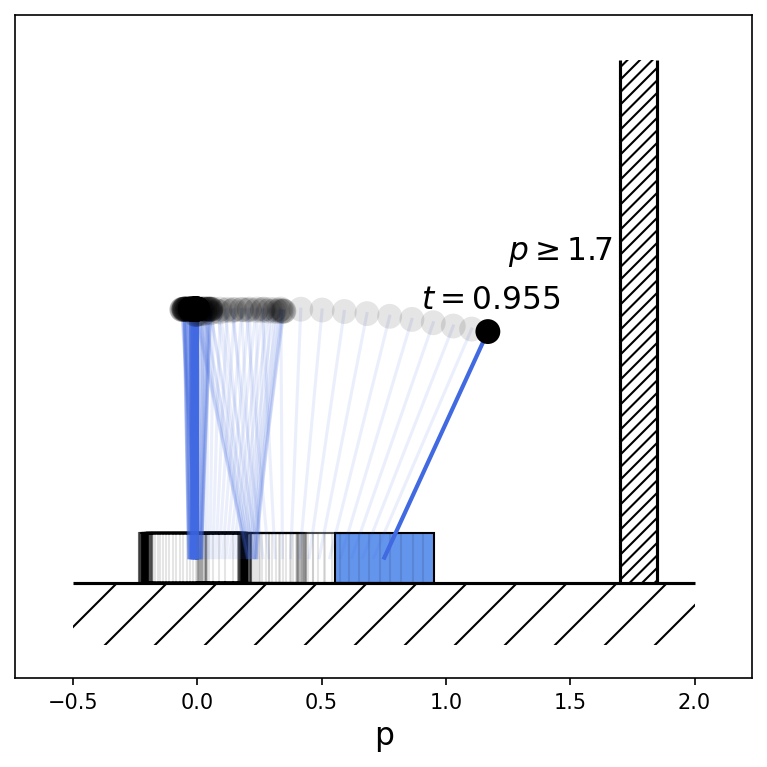} \label{fig:ex62} }
        \subfloat[]{ \hspace{-0.2cm}\includegraphics[width=0.48\linewidth]{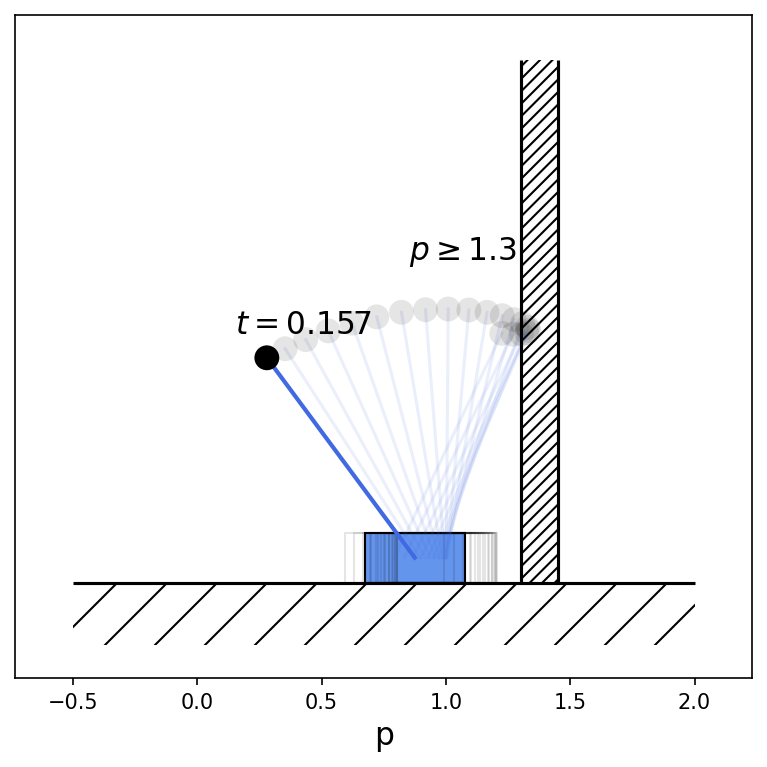} \label{fig:ex61} }
        \caption{Simulation of the Cart-pole system with LQR controller. (a) Convergence
        of the Cart-pole system under LQR control, with the estimated collision
        time bounds (red) and reach-time bounds (blue). (b) Motion snapshots within
        $T = [0.955,10.413]$ demonstrate that the system finally stabilizes if there
        is no collision occurred. (c) Motion snapshots within the estimated
        collision time interval $I_{U}^{1} = [0.157,0.955]$.}
        \label{fig:ex6}\vspace{-2em}
    \end{figure}


    We initialize the system with the state $X_{0}=(0.8, 0.3, -\pi/4, 0.2)$ and
    the target set as
    $X_{F}=\{ x \,|\, 10^{-4}\leq\|x\|_{\infty}\leq 10^{-2}, x \in X \}$, and apply
    our method to verify if a given LQR controller $\pi_{LQR}$ can stabilize the
    system without collision. The system has principal eigenvalues
    $\lambda_{1,2}= -0.633\pm0.473j$ and $\lambda_{3,4}=-0.752\pm3.997j$ and a
    stable equilibrium at the origin. We first set the unsafe set as $X_{U}^{1}=\{
    x \, | \,p+l\sin{(\theta)}> 1.3\}$, which represents a vertical wall. Our method
    outputs the reachable time bounds to the target set and the unsafe set as $I_{F}
    =[5.512,10.413]$ and $I_{U}^{1}=[0.157,0.955]$ as visualized in the Fig.
    \ref{fig:ex60}.

    When we change the unsafe set to
    $X_{U}^{2}=\{x \, | \,p+l\sin{(\theta)}> 1.7\}$, the corresponding estimated
    reachable time interval is $I_{U}^{2} = \emptyset$. Following corollary
    \ref{cor:main}, we can conclude that the unsafe set $X_{U}^{2}$ is unreachable
    from $X_{0}$ and the controller can safely drive the system to the equilibrium
    without collision with the wall. Fig. \ref{fig:ex61} depicts the simulation snapshots
    of the pole during the estimated unsafe time interval $I_{U}^{1}$ driven by the
    LQR controller. Indeed, the pole is shown to collide with the wall in this interval.

\end{example}




\begin{example}
    \label{ex6}(Multi-agent consensus system)
    Consider a multi-agent system consisting of $8$ quadrotors operating cooperatively in a shared workspace contains known signal interference regions. 
    The task requires concurrent arrival at a designated recovery area and prevent the simultaneous entry of all quadrotors into the signal interference zones to maintain a persistent connection with ground station.
    The agents' coordination is governed by a communication
    graph $\mathcal{G}$, as shown in Fig. \ref{fig:ex70}, with vertex set $\mathcal{V}=\{0,\cdots, n-1\}$ and edge set $\mathcal{E} =\{1,\cdots,m\}$. The neighboring nodes of agent $i$ is denoted by $N_{i} \subset \mathcal{V}$.
    The system state $q$ is represented by the augmented
    state vector $q = [x_{1},y_{1},\ldots,x_{n},y_{n}]$.
    Each component of the state, such as $x$, follows the nonlinear distributed
    consensus protocol given by

    \begin{align*}
        \dot{x}_{i}= u_{i} = -\sigma_{i}(x_{i}) \sum_{j\in N_i}a_{ij}(x_{i}-x_{j}), \forall i \in \mathcal{V}
    \end{align*}
    where $\sigma_{i}(\cdot)$ is continuous and positive for all $i \in \mathcal{V}$,
    and $a_{ij}(\cdot)$ is Lipschitz continuous for all $(i,j) \in \mathcal{E}$
    with $a_{ij}(x) >0$ when $x\neq0$ and $a_{ij}=0$. 
    We consider
    $\sigma_{i}(z)=\frac{1}{1+\exp{(-z)}}$ and
    $a_{ij}(\cdot)=r_{ij}[1.2\tanh{(\cdot)}+\sin{(\cdot)}]$ with randomly
    sampled $r_{ij}\in (0,1)$. 
    By adding a linear state feedback controller $u(q) = K q$ to right hand side of the system,
    we ensure that the origin is an asymptotically stable
    isolated equilibrium of this $16$-dimensional system which in turn ensures convergence of numerical methods used for computation of Koopman eigenfunctions. As shown in \cite{andreasson2014distributed}, each agent's state will converge to some equilibrium point $(x_{i}^{*},y_{i}^{*})$ if
    $\mathcal{G}$ is connected and undirected and each $a_{ij}(\cdot)$ is an odd function. We are interested in this example to verify the convergence of each agent's state to its equilibrium, which is $(x_{i}^{*},y_{i}^{*})=(0,0)$ in our settings, from arbitrary initial states.

    \begin{table}[!hb]
        \centering
        \aboverulesep = 0pt \belowrulesep = 0pt
        \begin{adjustbox}
            {width=0.48\textwidth,center,keepaspectratio}
            \begin{tabular}{|M{0.04\textwidth}|M{0.12\textwidth} M{0.12\textwidth}|}
                \toprule \multirow{2}{*}{Agent}   & \multicolumn{2}{c|}{$X_{0}^{i}$} \\
                \cmidrule{2-3}                    & $x_{i}(0)\in$                   & $y_{i}(0) \in$     \\
                \midrule \rowcolor[gray]{.90} $1$ & $[2.276, 2.476]$                & $[-1.229, -1.029]$ \\
                $2$                               & $[-1.754, -1.554]$              & $[1.784, 1.984]$   \\
                \rowcolor[gray]{.90} $3$          & $[1.946, 2.146]$                & $[-1.612, -1.412]$ \\
                $4$                               & $[-0.392, -0.192]$              & $[0.996, 1.196]$   \\
                \rowcolor[gray]{.90} $5$          & $[1.627, 1.827]$                & $[-1.759, -1.559]$ \\
                $6$                               & $[0.725, 0.925]$                & $[1.439, 1.639]$   \\
                \rowcolor[gray]{.90} $7$          & $[0.149, 0.349]$                & $[-1.776, -1.576]$ \\
                $8$                               & $[-2.422, -2.222]$              & $[-1.192, -0.992]$ \\
                \bottomrule
            \end{tabular}
            \begin{tabular}{|M{0.11\textwidth}|}
                \toprule Principal Eigenvalues \\
                \midrule $-0.442, -0.839$      \\
                $-0.819,-0.509$                \\
                $-0.719,-0.669$                \\
                $-0.645,-0.565$                \\
                $-0.442,-0.872$                \\
                $-0.779,-0.585$                \\
                $-0.705,-0.672$                \\
                $-0.627 \pm 0.004j$            \\
                \bottomrule
            \end{tabular}
        \end{adjustbox}
        \caption{Initial sets and principal eigenfunctions of the multi-agent
        consensus system.}
        \label{tab: multi-agent initial sets}
    \end{table}

\begin{figure}[H]
    \centering
    \begin{subfigure}[t]{0.45\linewidth}
    \centering
    \includegraphics[width=\linewidth]{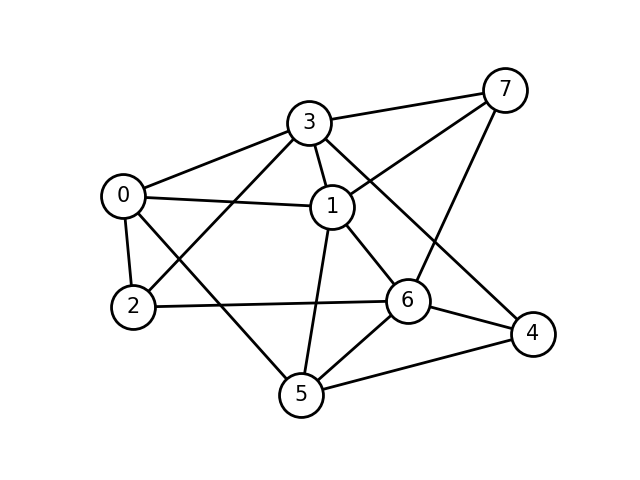}
    \caption{A consensus system with $8$ agents each with 2-dimensional states.}
    \label{fig:ex70}
    \end{subfigure}
    \hfill
    \begin{subfigure}[t]{0.53\linewidth}
    \centering
    \includegraphics[width=\linewidth]{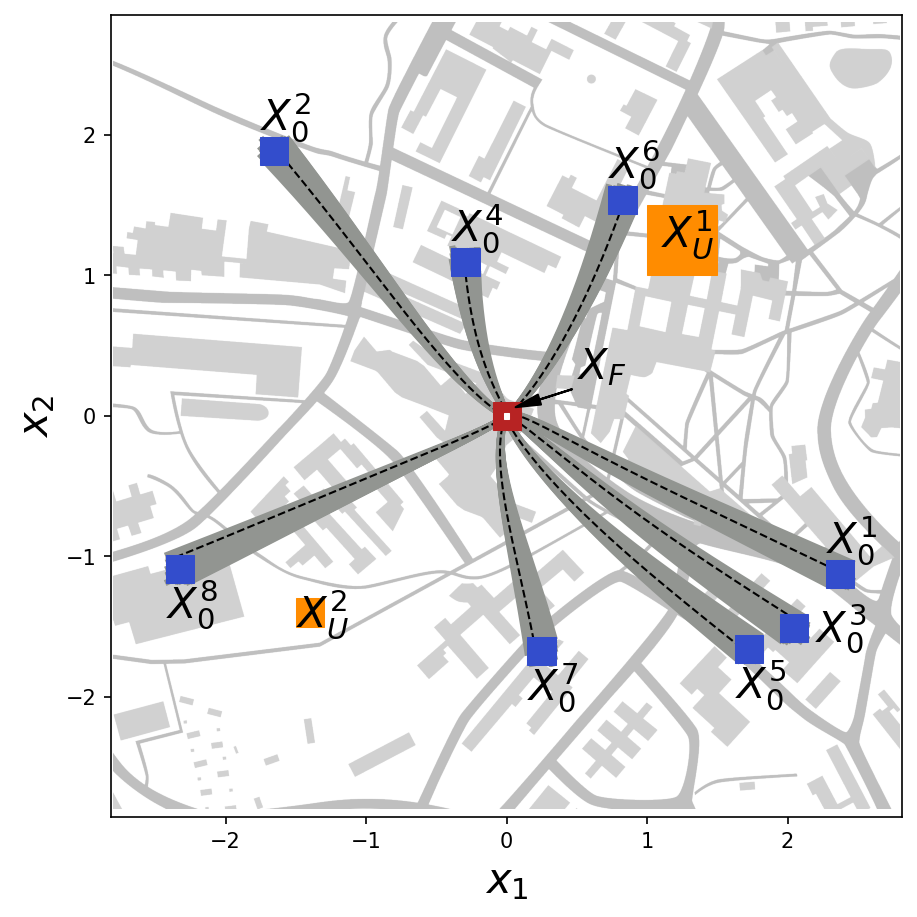}
    \caption{Reachable sets (in gray) of agents from corresponding initial
        sets for duration $T=[0,6.65]$ with simulated converging safe
        trajectories (black dashed lines).}
    \label{fig:ex71}
    \end{subfigure}
    \caption{Reach-time bounds estimation for the Multi-agent consensus system.}
\end{figure}

    As illustrated in Table \ref{tab: multi-agent initial sets}, we consider the
    cartesian product of $8$ random chosen sets as the initial set
    $X_{0} = \bigtimes_{i=1}^{8} X_{0}^{i}$. Furthermore, we define the target
    set as $X_{F}=\{q \,|\, 0.02 \leq \| q\|_{\infty}\leq 0.1, q \in X\}$ to
    verify each agent's state convergence around the equilibrium point $(0,0)$, as
    well as unsafe sets $X_{U}^{1} = \{q \,|\, \| q - 1.25 \|_{\infty}\leq 0.25, q
    \in X\}$ and $X_{U}^{2} = \{q \,| \,\| q + 1.4 \|_{\infty}\leq 0.1, q \in X
    \}$ 
    to denote signal interference regions.
    Under this configuration, our method estimates the reach time bound for
    $X_{F}$ as $I = [5.855, 6.650]$, and the reach-time bounds estimation for
    two unsafe sets are both $\emptyset$, implying that all agent avoid entering
    the unsafe regions while converging to
    the neighborhood $X_{F}$ around the equilibrium point. The simulated
    reachable sets of the system shown in Fig. \ref{fig:ex71} confirm the validity
    of our algorithm's results.
\end{example}

\begin{table*}
    [!htbp]
    \centering
    \aboverulesep = 0pt \belowrulesep = 0pt

    \begin{threeparttable}

    \begin{adjustbox}
        {width=\textwidth,center} \def\arraystretch{1.4} 
        \begin{tabular}{ |M{0.08\textwidth} |M{0.03\textwidth} | M{0.2\textwidth} | M{0.08\textwidth}  M{0.08\textwidth} | M{0.08\textwidth} M{0.11\textwidth} |M{0.11\textwidth}| M{0.10\textwidth} | M{0.12\textwidth}|}
            \toprule 
            \multicolumn{3}{|c|}{\textbf{Benchmark Problem}} & \multicolumn{7}{c|}{\textbf{Verification Result and Computation Time (s)}} \\
            \cmidrule{1-6} \cmidrule{7-10} 
            \midrule 
            \multirow{2}{*}{\textbf{System}} & \multirow{2}{*}{\textbf{Dim.}} & \multirow{2}{*}{\textbf{Benchmark ID}\tnote{e}} &  \multicolumn{2}{c|}{\textbf{Ours}} & \multicolumn{2}{c|}{\textbf{Set-Propagation}} & \textbf{Counterexample Searching} & \textbf{Hamilton-Jacobi PDE} & \textbf{Barrier Certificate \cite{prajna2004safety}}          \\
            \cmidrule{4-10} 
            & & & \textbf{Principal eigenpairs} & \textbf{Reach-time bounds} & \textbf{CORA} \cite{althoff2016cora} & \textbf{JuliaReach} \cite{bogomolov2019juliareach} & \textbf{dReach} \cite{kong2015dreach} & \textbf{hj\_reachability} & \textbf{SOSTOOLS} \cite{papachristodoulou2013sostools} \\
            \midrule                                                             
            \multirow{6}{*}{\textbf{\shortstack{Duffing's \\ oscillator \\(Example 1)}}} & \multirow{6}{*}{2} 
            & \gc DUFF\_REA\_BOX & \gc 13.42 & \gc Rea (0.07) & \gc Ukn (11.71) & \gc Ukn (13.70) & \gc $\delta-$Rea (93.12) & \gc Rea (57.54) & \gc Ukn (1.06) \\
            & & DUFF\_REA\_CVX\_LS & 15.07 & Rea (0.06) & N/A & N/A & $\delta-$Rea (3.13) & Rea (102.41) & Ukn (0.89) \\
            & & \gc DUFF\_REA\_NCVX\_LS & \gc 13.49 & \gc Rea (0.07) & \gc N/A & \gc N/A & \gc $\delta-$Rea (291.29) & \gc N/S \tnote{a} & \gc N/A \tnote{b}  \\
            \cmidrule{3-10} 
            & &  DUFF\_REA\_DSJT\_BOX & 17.83 & Rea (0.08) & N/A & N/A & Unr (0.12) & Rea (89.86) & Ukn (1.09)  \\
            & & \gc DUFF\_UNR\_BOX & \gc 13.95 & \gc Inc (0.07) & \gc Ukn (16.29) & \gc Ukn (14.32) & \gc T/O & \gc Unr (190.15) & \gc Unr (1.52) \\
            & & DUFF\_UNR\_CVX\_LS & 15.10 & Inc (0.08) & N/A & N/A & T/O & Unr (103.87) & Unr (0.88) \\
            \midrule
            \multirow{2}{*}{\textbf{\shortstack{Stable \\limit cycle}}} 
            & \multirow{2}{*}{2} 
            & \gc NL-SLC\_REA\_BOX & \gc \scriptsize Analytical  & \gc PRea (0.04) & \gc S/E & \gc Ukn (35.25) & \gc T/O & \gc Rea (1348.18) & \gc N/A \tnote{b} \\
            & & NL-SLC\_UNR\_BOX & \scriptsize Analytical & Unr (0.04) & S/E & Unr (37.72) & T/O & Unr (1375.69) & N/A \tnote{b} \\
            \midrule
            \multirow{2}{*}{\textbf{\shortstack{Saddle point\\system}}} 
            & \multirow{2}{*}{2} & \gc NL-EIG\_BWD\_REA\_BOX & \gc \scriptsize Analytical & \gc Rea (0.03) & \gc N/A & \gc N/A & \gc T/O & \gc Rea (33.76) & \gc N/A \tnote{b} \\
            & & NL-EIG\_BWD\_UNR\_BOX & \scriptsize Analytical & Unr (0.03) & N/A & N/A & T/O & Unr (33.54) & N/A \tnote{b} \\
            \midrule
            \multirow{3}{*}{\textbf{\shortstack{Cartpole \\ (Example 2)}}} 
            & \multirow{3}{*}{4} & \gc CP-LQR\_REA\_CVG & \gc 64.86 & \gc Rea (0.08) & \gc $\times$ (4.67) & \gc $\times$ (35.60) & \gc T/O & \gc OOM & \gc N/A \tnote{b}\\
            & & CP-LQR\_REA\_USF & 54.76 & Rea (0.11) & N/A & N/A & Rea (0.43) & OOM & N/A \tnote{b} \\
            & & \gc CP-LQR\_UNR\_SAF & \gc 21.16 & \gc Unr (0.04) & \gc N/A & \gc N/A & \gc T/O & \gc OOM & \gc N/A \tnote{b} \\
            \midrule
            \multirow{3}{*}{\textbf{\shortstack{Multi-agent\\consensus\\(Example 3)}}} 
            & \multirow{3}{*}{16} & MAS-CON\_REA\_BOX & 801.86 & Rea (0.06) & Rea (77.26) & Rea (1110.07) & T/O & OOM & N/A \tnote{b} \\
            & & \gc MAS-CON\_UNR\_SAF\_0 & \gc 782.91 & \gc Unr (0.05) & \gc Unr (76.03) & \gc Unr (1110.22) & \gc T/O & \gc OOM  & \gc N/A \tnote{b} \\
            & & MAS-CON\_UNR\_SAF\_1 &787.40 & Unr (0.06) & Unr (76.38) & Unr (1099.77) & T/O & OOM  & N/A \tnote{b} \\
            \bottomrule
        \end{tabular}
    \end{adjustbox}
    \begin{tablenotes}
            \item[a] The toolkit does not support general level set definition.
            \item[b] Not applicable due to trigonometric functions or fractional polynomials.
            \item[c] Rea --- Reachable; PRea --- Periodically reachable; Unr --- Unreachable; Ukn --- Unknown; Inc --- Inclusive. 
            \item[d]  N/A -- Not applicable; N/S --- Not support; S/E --- Set explosion; OOM --- Out of memory; T/O --- Time out; $\times$ --- Unsound result.
            \item[e] For more details regarding the benchmark specifications, please follow the shared link for the code. \protect\footnotemark
    \end{tablenotes}

    \caption{
    Comparative evaluation of reachability analysis methodologies on benchmarks. }
    \label{tab:benchmark}
    \end{threeparttable}
\end{table*}
\vspace{-2em}
\subsection{Comparative study}
As summarized in Table \ref{tab:benchmark}, our method distinguishes itself from existing reachability analysis techniques, each with distinct trade-offs. 
Set-propagation methods, for instance, scale well to high-dimensional systems and can reliably prove non-reachability. However, they are hampered by the wrapping effect, which leads to excessive conservatism over long horizons, and they generally cannot certify reachability. While tools like JuliaReach mitigate this conservatism to some extent using so-called ``lazy set representations", the fundamental limitation remains. 
In contrast, falsification-based tools like dReach and methods relying on HJB equations face their own challenges: the former struggles with high-dimensional systems and cannot prove non-reachability, while the latter is severely constrained by the curse of dimensionality. 
A different approach, barrier certificate synthesis via Sum-of-Squares (SOS) optimization, uniquely addresses infinite-horizon non-reachability but is often restricted to polynomial systems and cannot verify reachability. 
In contrast to these spatial or algebraic approaches, our method introduces a new paradigm by transforming the reachability problem into the temporal domain via the Koopman spectrum. This enables the reliable, infinite-horizon verification of diverse reachability properties (both reachability and non-reachability) and, through a sampling mechanism, extends its applicability to complex, non-convex or disjoint sets, thereby overcoming the primary limitations of the aforementioned techniques.

\footnotetext{\footnotesize https://github.com/Aalto-Nonlinear-Systems-and-Control/Time-to-Reach}
\section{Conclusion and Discussion}
This paper presents a methodology for verifying reachability based on the Koopman spectrum. The proposed method investigates the existence of feasible reachable time intervals to verify the reachability of systems, thereby avoiding the unnecessary computational overhead and conservatism arising from approximating the actual reachable sets.
Numerical experiments showcase our method's capability in verifying reachability between non-convex and even disconnected sets, over unbounded time horizons, and in the context of backward reachability. In addition, high-dimensional case studies illustrate the scalability of our method.
Since our method is based on the over-approximation of the exact reachable time range, the non-emptiness of the intersection of time-to-reach bounds does not serve as a sufficient basis for inferring reachability, although in such cases, the estimates can be combined with simulations to verify reachability. Our future research will focus on complementing the theory and seek its application when considering input control signals.\\







\section*{References}
\bibliographystyle{ieeetr}
\bibliography{refs}
\begin{IEEEbiographynophoto}{Jianqiang Ding} is a doctoral researcher at the Nonlinear Systems and Control Group of Aalto University, Espoo, Finland.
He received his M.Eng. degree in computer technology from Shenzhen University, China in 2020. 
His research interests are centered around formal methods, control theory, and their applications in robotics.
\end{IEEEbiographynophoto}

\begin{IEEEbiographynophoto} 
{Shankar A. Deka} (Member, IEEE) is an Assistant Professor of Automatic Control in the Department of Electrical Engineering and Automation at Aalto University, Finland. He received his PhD in Mechanical Science and Engineering from the University of Illinois at Urbana-Champaign, USA, in 2019.

He previously held postdoc positions at KTH
Royal Institute of Technology, Stockholm and
University of California, Berkeley between 2019 and 2023. His research focus is on nonlinear stability theory, robust and optimal control, and learning for dynamics and controls, with applications in medical robotics and precision agriculture.

Dr. Deka is on the editorial board Unmanned Systems, a board
member of IEEE CSS Finland Chapter, and affiliated to the Finnish
Center for Artificial Intelligence.
\end{IEEEbiographynophoto}


\end{document}